\documentclass[journal,twocolumn,10pt]{IEEEtran} 
\usepackage[]{times} 
\usepackage{epsfig} 
\usepackage{subfigure}
\usepackage{amsmath} 
\usepackage{amssymb}  
\usepackage{graphicx} 
\usepackage{subfigure}
\usepackage{multirow} 
\usepackage{url}
\renewcommand{\baselinestretch}{1} 
\newcommand{\bega}{\begin{eqnarray}}
\newcommand{\ega}{\end{eqnarray}}
\newcommand{\bb}{\begin{equation}}
\newcommand{\ee}{\end{equation}}

\newtheorem{te}{Theorem}
\newtheorem{lema}{Lemma}

\begin{document}
\title{Error Correction Capability of Column-Weight-Three LDPC Codes: Part II}
\author{Shashi~Kiran~Chilappagari,~\IEEEmembership{Student~Member,~IEEE,}~Dung~Viet~Nguyen,~\IEEEmembership{Student~Member,~IEEE,}~Bane~Vasic,~\IEEEmembership{Senior Member,~IEEE,}~and~Michael~W.~Marcellin,~\IEEEmembership{Fellow,~IEEE}
\thanks{Manuscript received \today. This work is funded by NSF under Grant CCF-0634969, ECCS-0725405, ITR-0325979 and by the INSIC-EHDR program.}
\thanks{S. K. Chilappagari, D. V. Nguyen, B. Vasic and M. W. Marcellin are with the Department of Electrical and Computer Engineering, University of Arizona, Tucson, Arizona, 85721 USA. (emails: \{shashic, nguyendv, vasic, marcellin\}@ece.arizona.edu.}
}
\markboth{Submitted to  IEEE Transactions on Information Theory, July 2008}%
{Submitted to IEEE Transactions on Information Theory, July 2008}
\maketitle
\begin{abstract}
\noindent The relation between the girth and the error correction capability of column-weight-three LDPC codes is investigated. Specifically, it is shown that the Gallager A algorithm can correct $g/2-1$ errors in $g/2$ iterations on a Tanner graph of girth $g \geq 10$.
\end{abstract}
\textbf{\small Keywords:} {\small Low-density parity-check codes, Gallager A algorithm, error floor, girth}
\section{Introduction}
Iterative message passing algorithms for decoding low-density parity-check (LDPC) codes \cite{gallager} operate by passing messages along the edges of a graphical representation of the code known as the Tanner graph \cite{tanner}. These algorithms are optimal when the underlying graph is a tree (see \cite{richardsonurbanke},\cite{richardsonurbankeshokrollahi} for general theory of LDPC codes), but in the presence of cycles, the decoding becomes sub-optimal and there exist low-weight patterns known as near codewords \cite{weaknessmackay} or trapping sets \cite{rich} uncorrectable by the decoder. It is now well established that the trapping sets lead to error floor in the high signal-to-noise (SNR) region (see \cite{breakingtrappingsets} for a list of references). While it is generally known that high girth codes have better performance in the error floor region, the exact relation between the girth and the slope of the frame error rate (FER) curve in the error floor region is unknown. 

In this paper, we consider transmission over the binary symmetric channel (BSC) and the corresponding hard decision decoding algorithms. We focus on column-weight-three codes which are of special importance as their decoders have very low complexity and are interesting for a wide range of applications. We show that a column-weight-three LDPC code with Tanner graph of girth $g \geq 10$ can correct all error patterns of weight $g/2-1$ or less under the Gallager A algorithm, thereby showing that there are no trapping sets with critical number less than $g/2$ (see \cite{colwtthreepaper} pp. 4-6, for missing definitions).  In \cite{breakingtrappingsets}, we showed that the slope of the FER curve in the error floor region is determined by the minimal critical number. In \cite{colwtthreepaper}, we proved that a column-weight-three LDPC code with Tanner graph of girth $g \geq 10$ always has a trapping set of size $g/2$ with critical number $g/2$ for the Gallager A algorithm. It is worth nothing that for codes with Tanner graphs of girth $g=6$ and $g=8$, girth alone cannot guarantee correction of  two and three errors, respectively. It can be easily shown that codes with girth six Tanner graphs can correct two errors if and only if the Tanner graph does not contain a codeword of weight four. For codes with Tanner graphs of girth eight, we established necessary and sufficient conditions to correct three errors \cite{itwpaper}. Thus, with the results presented here, the problem of determining the slope of the FER curve in the error floor region for column-weigh-three codes under the Gallager A algorithm is now completely solved.

The rest of the paper is organized as follows. In Section \ref{prelims}, we establish the notation, define the Gallager A algorithm and analyze the algorithm for the first $k$ ($g/4-1 \leq k < g/4$) iterations. In Section \ref{maintheorem}, we prove our main result and we conclude in Section \ref{discussion} with a few remarks.
\section{Preliminaries}\label{prelims}
\subsection{Notation}
Let $\mathcal{C}$ be an LDPC code with column weight three and length $n$. The Tanner graph $G$ of $\mathcal{C}$ consists of two sets of nodes: the set of variable nodes $V$ with $|V|=n$ and the set of check nodes $C$. The check nodes (variable nodes) connected to a variable node (check node) are referred to as its neighbors. An edge $e$ is an unordered pair $\{v,c\}$ of a variable node $v$ and a check node $c$ and is said to be incident on $v$ and $c$. A directed edge $\vec{e}$ is an ordered pair $(v,c)$ or $(c,v)$ corresponding to the edge $e=\{v,c\}$. With a moderate abuse of notation, we denote directed edges by simple letters (without arrows) but specify the direction. The girth $g$ is the length of the shortest cycle in $G$. For a given node $u$, the neighborhood of depth $d$, denoted by $\mathcal{N}^d_u$, is the induced subgraph consisting of all nodes reached and edges traversed by paths of length at most $d$ starting from $u$ (including $u$). The directed neighborhood of depth $d$ of a directed edge $e=(v,c)$ denoted by $\mathcal{N}^d_{e}$, is defined as the induced subgraph containing all edges and nodes on all paths $e_1,\ldots,e_d$ starting from $v$ such that $e_1\neq e$ (see \cite{richardsonurbanke} for definitions and notation). In a Tanner graph with girth $g$, we note that $\mathcal{N}^{t}_u$ is a tree when $t \leq g/2-1$. Also, if $e_1=(v,c)$ and $e_2=(c,v)$, then $\mathcal{N}^i_{e_1} \cap \mathcal{N}^j_{e_2} = \phi$ for $i+j<g-1$. Let $k$ denote the number of independent iterations as defined in \cite{gallager}. The original value of a variable node is its value in the transmitted codeword. We say a variable node is good if its received value is equal to its original value and bad otherwise. A message is said to be correct if it is equal to the original value of the corresponding variable node and incorrect otherwise. In this paper, $\circ$ denotes a good variable node, $\bullet$ denotes a bad variable node and $\square$ denotes a check node. For output symmetric channels (see \cite{richardsonurbanke}), without loss of generality, we can assume that the all zero codeword is transmitted. We make this assumption throughout the paper. Hence, a bad variable node has received value $1$ and an incorrect message has a value of $1$. A configuration of bad variable nodes is a subgraph in which the location of bad variables relative to each other is specified. A valid configuration $\mathcal{C}_g$ is a configuration of at most $g/2-1$ bad variable nodes free of cycles of length less than $g$. The set of bad variable nodes in $\mathcal{N}^d_e$ is denoted by $\mathcal{B}(\mathcal{N}^d_e)$ and $|\mathcal{B}(\mathcal{N}^d_e)|$ is denoted by $B(\mathcal{N}^d_e)$. The number of bad variable nodes at depth $d$ in $\mathcal{N}^d_e$ is denoted by $b_e^d$.
\subsection{Hard Decision Decoding Algorithms}
Gallager in \cite{gallager} proposed two simple binary message passing algorithms for decoding over the BSC; Gallager A and Gallager B. See \cite{shokrollahi} for a detailed description of the Gallager B algorithm. For column-weight-three codes, which are the main focus of this paper, these two algorithms are the same. Every round of message passing (iteration) starts with sending messages from variable nodes to check nodes (first half of the iteration) and ends by sending messages from check nodes to variable nodes (second half of the iteration). Let $\mathbf{r}$, a binary $n$-tuple be the input to the decoder. Let $\omega_j(v,c)$ denote the message passed by a variable node $v$ to its neighboring check node $c$ in $j^{th}$ iteration and $\varpi_j(c,v)$ denote the message passed by a check node $c$ to its neighboring variable node $v$. Additionally, let $\omega_j(v,\colon)$ denote the set of all messages from $v$, $\omega_j(v,\colon \backslash c)$ denote the set of messages from $v$ to all its neighbors except to $c$ and $\omega_j(\colon,c)$ denote the set of all messages to $c$. Let $|\varpi(:,v)=m|$ denote the number of incoming messages to $v$ which are equal to $m \in \{0,1\}$. The terms $\omega_j(\colon \backslash v,c),\varpi_j(c,\colon),\varpi_j(c,\colon \backslash v), \varpi_j(\colon,v)$ and $\varpi_j(\colon \backslash c,v)$ are defined similarly. The Gallager A algorithm can then be defined as follows.
\begin{eqnarray}
\omega_1(v,c)&=&\mathbf{r}(v) \nonumber \\ 
\omega_j(v,c) &=& \left \{ \begin{array}{cl}
												1, & \mbox{if } \varpi_{j-1}(\colon \backslash c,v)=\{1\}\\
											0, & \mbox{if } \varpi_{j-1}(\colon \backslash c,v)=\{0\}\\
											\mathbf{r}(v), & \mbox{otherwise}  \end{array}\right. \nonumber \\
\varpi_j(c,v)&=& \left(\sum \omega_j(\colon \backslash v,c)\right) \mbox{mod } 2 \nonumber
\end{eqnarray}

At the end of each iteration, an estimate of each variable node is made based on the incoming messages and possibly the received value. In this paper, we assume that the estimate of a variable node is taken to be the majority of the incoming messages (see \cite{colwtthreepaper} for details).  The decoded word at the end of the $j^{th}$ iteration is denoted as $\mathbf{r}^{(j)}$. The decoder is run  until a valid codeword is found or a maximum number of iterations $M$ is reached, whichever is earlier. The output of the decoder is either a codeword or $\mathbf{r}^{(M)}$.
\begin{figure*}
\centering
\subfigure[] 
{
    \label{iteration2}

\includegraphics[height=2.4in]{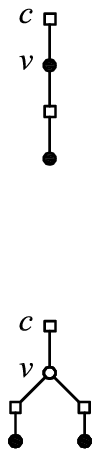}
}
\subfigure[] 
{
    \label{iteration3}

\includegraphics[height=2.4in]{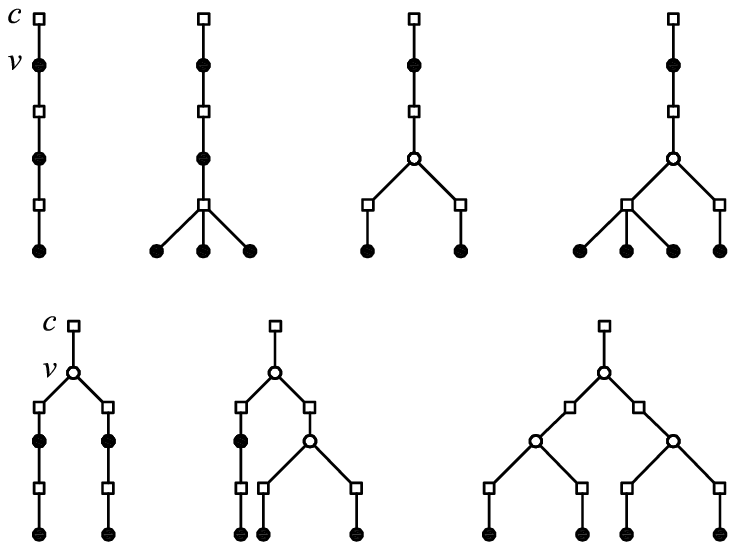}
}
\subfigure[] 
{
    \label{iteration4}

\includegraphics[height=3.5in]{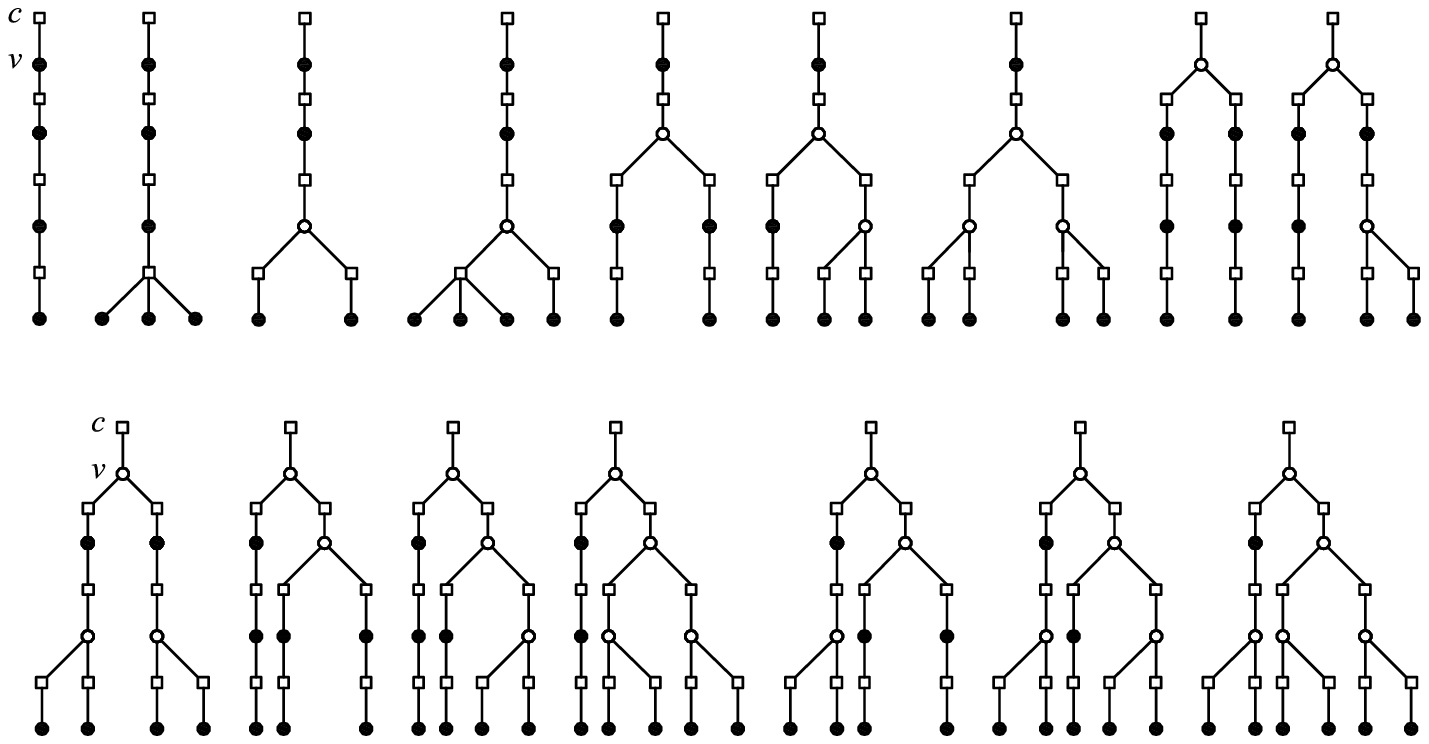}
}
\caption{Possible configurations of at most $2k+1$ bad variable nodes in the neighborhood of a variable node $v$ sending an incorrect message to check node $c$ in the $(k+1)^{th}$ iteration for \subref{iteration2} $k=1$, \subref{iteration3} $k=2$ and \subref{iteration4} $k=3$. }
\label{iterations234}
\end{figure*}

\subsection{The first $k$ iterations}\label{kiterations}
We begin with a lemma describing the messages passed by the Gallager A algorithm in a column-weight-three code.
\begin{lema}

(i) If $v$ is a bad variable node, then we have $\omega_1(v,:)=\{1\}$ and
\begin{itemize} 
\item $\omega_j(v,:)=\{1\}$ if $|\varpi_{j-1}(:,v)=1| \geq 2$, i.e., $v$ sends incorrect messages to all its neighbors if it receives two or more incorrect messages from its neighboring checks in the previous iteration.
\item $\omega_j(v,:\backslash c)=\{1\}$ and $\omega_j(v,c)=0$ if $\varpi_{j-1}(:\backslash c,v)=\{0\}$ and $\varpi_{j-1}(c,v)=1$, i.e., $v$ sends one correct message and two incorrect messages if it receives one incorrect message from its neighboring checks in the previous iteration. The correct message is sent along the edge on which the incorrect message is received. 
\item $\omega_j(v,:)=\{0\}$ if $\varpi_{j-1}(:,v)=\{0\}$, i.e., $v$ sends all correct messages if it receives all correct messages from its neighboring checks in the previous iteration.
\end{itemize}

(ii) If $v$ is a good variable node, then we have $\omega_1(v,:)=\{0\}$ and
\begin{itemize} 
\item $\omega_j(v,:)=\{0\}$ if $|\varpi_{j-1}(:,v)=0| \geq 2$, i.e., $v$ sends all correct messages  if it receives two or more  correct messages from its neighboring checks in the previous iteration.
\item $\omega_j(v,:\backslash c)=\{0\}$ and $\omega_j(v,c)=1$ if $\varpi_{j-1}(:\backslash c,v)=\{1\}$ and $\varpi_{j-1}(c,v)=0$, i.e., $v$ sends one incorrect message and two correct messages if it receives two incorrect messages from its neighboring checks in the previous iteration. The incorrect message is sent along the edge on which the correct message is received. 
\item $\omega_j(v,:)=\{1\}$, if $\varpi_{j-1}(:,v)=\{1\}$, i.e., $v$ sends all incorrect messages  if it receives all incorrect messages from its neighboring checks in the previous iteration.
\end{itemize}

(iii) For a check node $c$, we have,
\begin{itemize}
\item $\varpi_{j}(c,v)=\omega_j(v,c)\oplus 1$, if $|\omega_j(:,c)=1|$ is odd, i.e., $c$ sends incorrect messages along the edges on which it received correct messages and correct messages along the edges on which it received incorrect messages, if the total number of incoming incorrect messages from its neighboring variable nodes is odd.
\item $\varpi_{j}(c,v)=\omega_j(v,c)$, if $|\omega_j(:,c)=1|$ is even, i.e., $c$ sends incorrect messages along the edges on which it received incorrect messages and correct messages along the edges on which it received correct messages, if the total number of incoming incorrect messages from its neighboring variable nodes is even.
\end{itemize}

(iv) A variable node is estimated incorrectly at the end of an iteration if it receives at least two incorrect messages.
\end{lema}
\begin{proof}
Follows from the description of the Gallager A algorithm.
\end{proof}

Now let $v$ be a variable node which sends an incorrect message along the edge $e=(v,c)$ in the $(k+1)^{th}$ iteration. The message along $e$ depends only on the variable nodes and check nodes in $\mathcal{N}^{2k}_{e}$. Under the assumption that $\mathcal{N}^{2k}_{e}$ is a tree, the above observations provide a method to find all the possible configurations of bad variable nodes in $\mathcal{N}^{2k}_{e}$. We have the following two cases:

(a) $v$ is a bad variable node: In this case, there must be at least one variable node in $\mathcal{N}^2_{e}$ which sends an incorrect message in the $k^{th}$ iteration.

(b) $v$ is a good variable node: In this case, there must be at least two variable nodes in $\mathcal{N}^2_{e}$ which send an incorrect message in the $k^{th}$ iteration.

This is repeated $k$ times until we reach the first iteration, at which point only the nodes and edges in $\mathcal{N}^{2k}_{e}$ would have been explored and all of these are guaranteed to be distinct as $\mathcal{N}_e^{2k}$ is a tree. Since only bad variables send incorrect messages in the first iteration, we can calculate the number of bad variables in each configuration. Specifically, let $v$ be a variable node which sends an incorrect message along the $e=(v,c)$ in the second iteration. If $v$ is a good variable node, then $\mathcal{N}^2_{e}$ must have at least two bad variable nodes. If $v$ is bad, $\mathcal{N}^2_{e}$ must have at least one bad variable node. Following this approach, we have Fig. \ref{iteration2}, Fig. \ref{iteration3} and Fig. \ref{iteration4} which show the possible configurations of bad variable nodes in $\mathcal{N}^{2k}_e$ so that $v$ sends an incorrect message in the $(k+1)^{th}$ iteration, for $k=1, k=2$ and $k=3$, respectively.

\textit{Remarks:} (i) Fig. \ref{iterations234} shows configurations with at most $2k+1$ bad variable nodes. \\
(ii) We do not illustrate configurations in which a bad variable node receives two incorrect messages in the $k^{th}$ iteration, so that it sends an incorrect message in the $(k+1)^{th}$ iteration. However, all such configurations can be found by considering configurations involving good variable nodes and converting a good varaible node to a bad one. This increases the number of bad variable nodes in the condiguration. As will be seen later, such configurations are not relevant for establishing our main result.

The above observations help establish bounds on $B(\mathcal{N}_e^{2k})$, which we state in the following lemma.
\begin{lema}
(i) If $v$ is a bad variable node sending an incorrect message on $e=(v,c)$ in the $(k+1)^{th}$ iteration and $\mathcal{N}^{2k}_e$ is a tree, then $B(\mathcal{N}_e^{2k}) \geq k+1$. If $B(\mathcal{N}^{2k-2}_{e})=1$, i.e., $b^{d}_{e}=0$ for $d=2,4,\ldots,2(k-1)$, then $B(\mathcal{N}_e^{2k}) \geq 2^{(k-1)}+1$. If $B(\mathcal{N}_e^{2k-2})=2$, then $B(\mathcal{N}_e^{2k}) \geq 2^{(k-2)}+2$. 

(ii) If $v$ is a good variable node sending an incorrect message on $e=(v,c)$ in the $(k+1)^{th}$ iteration and $\mathcal{N}^{2k}_e$ is a tree, then $B(\mathcal{N}_e^{2k}) \geq 2k$. If $B(\mathcal{N}_e^{2k-2})=0$, then $B(\mathcal{N}^{2k}_{e})\geq 2^k$. If $B(\mathcal{N}_e^{2k-2})=1$, then $B(\mathcal{N}^{2k}_e)\geq 2^{(k-1)}+2^{(k-2)}+1$. If $B(\mathcal{N}_e^{2k-2})=2$, then $B(\mathcal{N}_e^{2k}) \geq 2^{(k-1)}+2$. 
\end{lema}
\begin{proof}
The proof is by induction on $k$. It is easy to verify the bounds for $k=2$. Let the bounds be true for some $k \geq 2$. Let $v_0$ be a bad variable node sending an incorrect message on $e=(v_0,c)$ in the $(k+2)^{th}$ iteration. Further, assume that $\mathcal{N}^{2k+2}_e$ is a tree. Then, $\mathcal{N}^1_{e}$ has at least one check node $c_1$ which sends an incorrect message along the edge $e_1=(c_1,v_0)$ in the $(k+1)^{th}$ iteration. This implies that $\mathcal{N}^2_e$ has at least one variable node $v_2$ sending an incorrect message in the $(k+1)^{th}$ iteration along the edge $e_2=(v_2,c_1)$. Since a path of length $2$ exists between $v_0$ and $v_1$, $\mathcal{N}^{i}_{e_2} \subset \mathcal{N}^{i+2}_{e}$. 

If $v_2$ is a bad variable node, then $B(\mathcal{N}^{2k}_{e_2}) \geq k+1$ and consequently $B(\mathcal{N}^{2k+2}_{e}) \geq k+2$. If $v_2$ is a good variable node, then $B(\mathcal{N}^{2k}_{e_2}) \geq 2k$ and consequently $B(\mathcal{N}^{2k+2}_{e}) \geq 2k+1 > k+1$. 

If $b^{d}_{e}=0$ for $d=2,4,\ldots,2k$, then $v_2$ is a good variable node such that $B(\mathcal{N}^{2k-2}_{e_2})=0$ which implies that  $B(\mathcal{N}^{2k}_{e_2})\geq 2^k$ by the induction hypothesis. Hence, $B(\mathcal{N}^{2k+2}_{e})\geq 2^k+1$. 

If $B(\mathcal{N}_e^{2k})=2$ then either (a) $v_2$ is a bad variable node with $b^{d}_{e_2}=0$ for $d=2,4,\ldots,2(k-1)$ which implies that $B(\mathcal{N}_{e_2}^{2k}) \geq 2^{(k-1)}+1$ by the induction hypothesis. Hence, $B(\mathcal{N}_e^{2k+2}) \geq 2^{(k-1)}+2$, or (b) $v_2$ is a good variable node with $B(\mathcal{N}_{e_2}^{2k-2})=1$ which implies that $B(\mathcal{N}^{2k}_{e_2})\geq 2^{(k-1)}+2^{(k-2)}+1$ by the induction hypothesis. Hence, $B(\mathcal{N}_e^{2k+2}) \geq 2^{(k-1)}+2^{(k-2)}+2 > 2^{(k-1)}+2$.

By the principle of mathematical induction, the bounds are true for all $k$ when $v_0$ is a bad variable node. The proofs are similar for the case when $v_0$ is a good variable node.
\end{proof}
\section{The Main Theorem}\label{maintheorem}
In this section, we prove that a column-weight-three code with Tanner graph of girth $g \geq 10$ can correct $g/2-1$ errors in $g/2$ iterations of the Gallager A algorithm. The proof proceeds by finding, for a particular choice of $k$, all configurations of $g/2-1$ or less bad variable nodes which do not converge in $k+1$ iterations and then prove that these configurations also converge in subsequent iterations. When $g/2$ is even, we use $k=g/4-1$ (or $g/2-1=2k+1$) and when $g/2$ is odd, we use $k=(g-2)/4$ (or $g/2-1=2k$). We deal with these cases separately.

\subsection{$g/2$ is even}
Let $v_0$ be a variable node which receives two incorrect messages along the edges $e_1=(c_1^1,v_0)$ and $e_2=(c_1^2,v_0)$ at the end of $(k+1)^{th}$ iteration. This implies that $N^1_{e_1}$ and $N^1_{e_2}$ each has a variable node, $v_2^1$ and $v_2^2$ respectively, that sends an incorrect message in the $(k+1)^{th}$ iteration to check nodes $c_1^1$ and $c_1^2$, respectively. Let $e_3=(v_2^1,c_1^1)$, $e_4=(v_2^2,c_1^2)$, $e_5=(c_1^1,v_2^1)$, and $e_6=(c_1^2,v_2^2)$ (see Fig. \ref{girth16_1} for an illustration). All possible configurations of bad variable nodes in $\mathcal{N}^{2k}_{e_3}$ and $\mathcal{N}^{2k}_{e_4}$ can be determined using the method outlined in Section \ref{kiterations}. Since there exists a path of length $3$ between $v_2^2$ and $c_1^1$, we have $\mathcal{N}^{i}_{e_4} \subset \mathcal{N}^{i+3}_{e_5}$. Also, $\mathcal{N}^{i}_{e_3} \cap \mathcal{N}^{j}_{e_5} =\phi $ for $i+j < g-1 = 4k+3$. Therefore, $\mathcal{N}^{i}_{e_3} \cap \mathcal{N}^{j}_{e_4} = \phi$ for $i+j < 4k$. This implies that $\mathcal{N}^{2k}_{e_3}$ and $\mathcal{N}^{2k}_{e_4}$ can have a common node only at depth $2k$. The total number of bad variable nodes in $\mathcal{N}^{2k}_{e_3} \cup \mathcal{N}^{2k}_{e_4}$, $B(\mathcal{N}^{2k}_{e_3} \cup \mathcal{N}^{2k}_{e_4})$, in any configuration is therefore lower bounded by $B(\mathcal{N}^{2k-2}_{e_3})+B(\mathcal{N}^{2k-2}_{e_4})+ \max(b^{2k}_{e_3},b^{2k}_{e_4})$ or equivalently $\max\left(B(\mathcal{N}^{2k-2}_{e_3})+B(\mathcal{N}^{2k}_{e_4}),B(\mathcal{N}^{2k}_{e_3})+B(\mathcal{N}^{2k-2}_{e_4})\right)$. We are interested only in the valid configurations, i.e., at most $g/2-1$ bad variable nodes, free from cycles of length less than $g$. We divide the discussion into three parts: (1) we find all the possible valid configurations for the case when $g=16$; (2) we then proceed iteratively for $g > 16$; (3) We consider the case $g=12$ separately as the arguments for $g\geq 16$ do not hold for this case.

\subsubsection{$g=16$}
Let $v$ be a variable node which sends an incorrect message in the iteration $k+1=g/4=4$ along edge $e=(v,c)$, given that there are at most seven bad variables and $\mathcal{N}^7_{e}$ is a tree. Fig. \ref{iteration4} illustrates different configurations of bad variable nodes in $\mathcal{N}^{6}_{e}$. As remarked earlier, Fig. \ref{iteration4} does not show configurations in which a bad variable node has to receive two incorrect messages in an iteration to send an incorrect message along the third edge in the next iteration. It can be seen in the proof of Theorem \ref{theorem2} that these cases do not arise in valid configurations. 

Let $v_0,v_2^1,v_2^2,e_1,e_2,e_3,e_4$ be defined as above with $k=3$. Using the arguments outlined above (and the constraint that $g=16$), all possible configurations such that $B(\mathcal{N}^{2k}_{e_3} \cup \mathcal{N}^{2k}_{e_4}) \leq 7$ can be found. Fig. \ref{girth16} shows all such possible configurations. 
\begin{figure}
\subfigure[] 
{
    \label{girth16_1}

\includegraphics[width=0.93in]{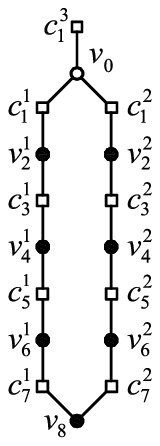}
}
\subfigure[] 
{
    \label{girth16_2}

\includegraphics[width=0.93in]{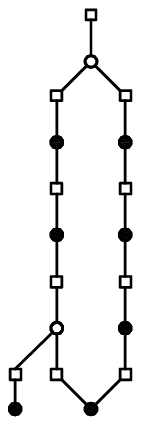}
}
\subfigure[] 
{
    \label{girth16_3}

\includegraphics[width=0.93in]{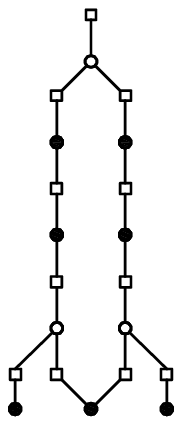}
}
\caption{Configurations of at most $7$ bad variable nodes, free of cycles of length  less than $16$, which do not converge in $4$ iterations.}
\label{girth16}
\end{figure}

\subsubsection{$g \geq 20$}
Let $\mathcal{C}_g$ be a valid configuration in which there exists a variable node $v_0$ which receives two incorrect messages along the edges $e_1=(c_1^1,v_0)$ and $e_2=(c_1^2,v_0)$ at the end of $(k+1)^{th}$ iteration. This implies that $\mathcal{N}^1_{e_1}$ and $\mathcal{N}^1_{e_2}$ each has a variable node, $v_2^1$ and $v_2^2$, respectively, that sends an incorrect message in the $(k+1)^{th}$ iteration. We have the following lemma.
\begin{lema}\label{lemma1}
$v_2^1$ and $v_2^2$ are bad variable nodes.
\end{lema}
\begin{proof}
The proof is by contradiction. We know that the total number of bad variables in any configuration is lower bounded by \\ $\max\left(B(\mathcal{N}^{2k-2}_{e_3})+B(\mathcal{N}^{2k}_{e_4}),B(\mathcal{N}^{2k}_{e_3})+B(\mathcal{N}^{2k-2}_{e_4})\right)$ and that $k>3$. We have two cases.

(a) $v_2^1$ and $v_2^2$ are both good variable nodes: We first note that in any valid configuration, $B(\mathcal{N}^{2k-2}_{e_3}) \geq 2$. Otherwise, we have $B(\mathcal{N}_{e_3}^{2k-2})=1$, and from Lemma 2, $B(\mathcal{N}^{2k}_{e_3}) \geq 2^{(k-1)}+2^{(k-2)} +1 >2k+1$ or, we have $B(\mathcal{N}_{e_3}^{2k-2})=0$, and $B(\mathcal{N}^{2k}_{e_3}) \geq 2^k +2 >2k+1$. Both cases are a contradiction as we have at most $2k+1$ bad variable nodes. Hence,  $B(\mathcal{N}^{2k-2}_{e_3}) \geq 2$.

Now, $B(\mathcal{N}^{2k}_{e_4})\geq 2k$, and hence we have $B(\mathcal{N}_{e_3}^{2k} \cup \mathcal{N}_{e_4}^{2k}) \geq B(\mathcal{N}^{2k}_{e_4}) + B(\mathcal{N}^{2k-2}_{e_3}) \geq 2k+2$, which is a contradiction.

(b) $v_2^1$ is a bad variable node and $v_2^2$ is a good variable node. The opposite case is identical. 

First we claim that in any valid configuration, $B(\mathcal{N}^{2k-2}_{e_3})\geq 2$. Since $v_2^1$ is a bad variable node, $B(\mathcal{N}^{2k-2}_{e_3}) \neq 0$. Assume that $B(\mathcal{N}^{2k-2}_{e_3})=1$. Then $B(\mathcal{N}^{2k}_{e_3})\geq 2^{(k-1)}+1$. Again,  $B(\mathcal{N}^{2k-2}_{e_4}) \geq 2$ implies that $B(\mathcal{N}^{2k}_{e_3} \cup \mathcal{N}^{2k}_{e_4}) \geq 2^{(k-1)}+3 > 2k+1$ (as $k>3$), which is a contradiction. Hence,  $B(\mathcal{N}^{2k-2}_{e_3})\geq 2$. 

Now, $B(\mathcal{N}^{2k-2}_{e_3})\geq 2$ and $B(\mathcal{N}^{2k}_{e_4})\geq 2k$, implies that $B(\mathcal{N}^{2k}_{e_3} \cup \mathcal{N}^{2k}_{e_4}) \geq 2k+2$ which is a contradiction.

Hence, $v_2^1$ and $v_2^2$ are both bad variable nodes.

\end{proof}
We now have the following theorem:
\begin{te}\label{theorem1}
If $\mathcal{C}_g$ is a configuration which does not converge in $(k+1)$ iterations, then there exists a configuration $\mathcal{C}_{g-4}$ which does not converge in $k$ iterations.
\end{te}
\begin{proof}
In the configuration $\mathcal{C}_g$, $v_2^1$ and $v_2^2$ are bad variable nodes which send incorrect messages to check nodes $c_1^1$ and $c_1^2$, respectively, in the $(k+1)^{th}$ iteration. This implies that $\mathcal{N}^1_{e_3}$ and $\mathcal{N}^1_{e_4}$ each has a check node, $c_3^1$ and $c_3^2$, respectively, that sends an incorrect message in $k^{th}$ iteration to $v_2^1$ and $v_2^2$, respectively. Now consider a configuration $\mathcal{C}_{g-4}$ constructed from $\mathcal{C}_g$ by removing the nodes $v_2^1,v_2^2,c_1^1,c_1^2$ and the edges connecting them to their neighbors and introducing the edges $(v_0,c_3^1)$ and $(v_0,c_3^2)$ (see Fig. \ref{configg_g4}). If $\mathcal{C}_g$ has at most $2k+1$ bad variable nodes and no cycles of length less than $g$, then $\mathcal{C}_{g-4}$ has at most $2(k-1)+1$ bad variable nodes and no cycles of length less than $g-4$. In $\mathcal{C}_{g-4}$ variable node $v_0$ receives two incorrect messages at the end of $k$ iterations and hence $\mathcal{C}_{g-4}$ is a valid configuration which does not converge in $k$ iterations.
\end{proof}

Theorem \ref{theorem1} gives a method to construct valid configurations of bad variable nodes for girth $g$ from valid configurations for girth $g+4$. Also, if $\mathcal{C}^1_{g}$ and $\mathcal{C}^2_{g}$ are two distinct valid configurations, then the configurations $\mathcal{C}^1_{g-4}$ and $\mathcal{C}^2_{g-4}$ constructed from $\mathcal{C}^1_{g}$ and $\mathcal{C}^2_{g}$, respectively, are distinct. Hence, the number of valid configurations for girth $g$ is greater than or equal to the number of valid configurations for girth $g+4$. Note that the converse of Theorem \ref{theorem1} is not true in general. However, for $g \geq 16$, we will show in Theorem \ref{theorem2} that any configuration for girth $g$ can be extended to a configuration for girth $g+4$.
\begin{te}\label{theorem2}
For $g/2$ even and $k \geq 3$, there are only three valid configurations which do not converge in $(k+1)$ iterations.
\end{te}
\begin{proof}
For $k=3$, we have $g=16$ and there are only three valid configurations as given in Fig \ref{girth16}. So, for $g \geq 16$ and $g/2$ even, there can be at most three valid configurations. Each valid configuration for $g=16$, can be extended to a configuration $\mathcal{C}_{20}$ for $g=20$  by the addition of two bad variable nodes $v_{new}^1$ and $v_{new}^2$ in the following way. Remove the edges $(v_0,c_1^1)$ and $(v_0,c_1^2)$. Add bad variable nodes $v_{new}^1$ and $v_{new}^2$ and check nodes $c_{new}^1$ and $c_{new}^2$. Introduce the edges $(v_0,c_{new}^1)$, $(v_0,c_{new}^2)$, $(v_{new}^1, c_{new}^1)$, $(v_{new}^2,c_{new}^2)$, $(v_{new}^1,c_1^1)$ and $(v_{new}^2,c_{1}^2)$ (see Fig. \ref{configg_g4} for an illustration). It can be seen that $\mathcal{C}_{20}$ is a valid configuration for girth $g=20$. In general, the configurations constructed using the above method from the valid configurations for $g \geq 16$ are valid configurations for $g+4$. Fig. \ref{girth20} illustrates the three configurations for all $k \geq 3$.
\end{proof}
\textit{Remark:} In all valid configurations $\mathcal{C}_g$ with $g \geq 16$, no bad variable node receives two incorrect messages at the end of the $(k+1)^{th}$ iteration.
\begin{figure}
\begin{centering}
\subfigure[] 
{
    \label{configg_1}

\includegraphics[width=2in]{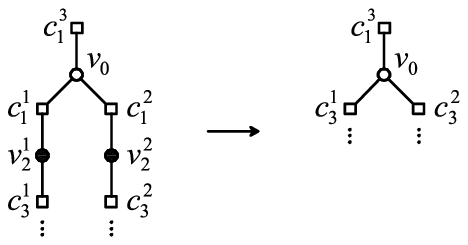}
}
\subfigure[] 
{
    \label{configg_g4}

\includegraphics[width=2in]{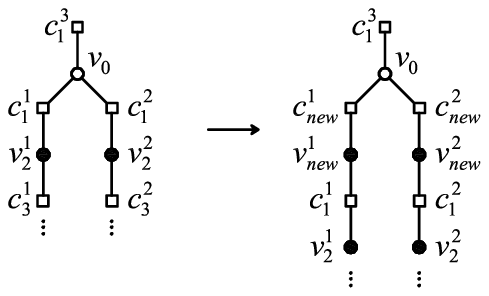}
}
\caption{\subref{configg_1} Construction of $\mathcal{C}_{g-4}$ from $\mathcal{C}_{g}$, \subref{configg_g4} Construction of $\mathcal{C}_{g+4}$ from $\mathcal{C}_{g}$.}
\label{configgtog_4}
\end{centering}
\end{figure}
\begin{figure}
\subfigure[] 
{
    \label{girth20_1}

\includegraphics[width=0.93in]{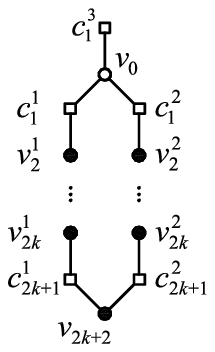}
}
\subfigure[] 
{
    \label{girth20_2}

\includegraphics[width=0.93in]{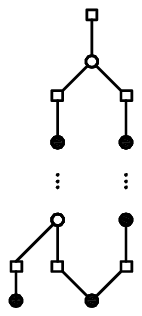}
}
\subfigure[] 
{
    \label{girth20_3}

\includegraphics[width=0.93in]{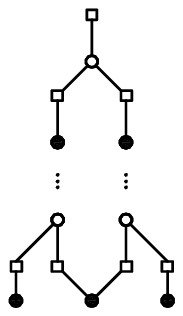}
}
\caption{Configurations of at most $2k+1$ bad variable nodes free of cycles of length less than $4k+4$ which do not converge in $(k+1)$ iterations.}
\label{girth20}
\end{figure}
\begin{te}\label{theorem3}
All valid configurations $\mathcal{C}_g$ converge to a codeword in $g/2$ iterations.
\end{te}
\begin{proof}
We prove the theorem for one configuration for $g=16$ only. The proof is similar for other configurations. 
At the end of fourth iteration, let $v_0$ receive two incorrect messages (see Fig \ref{girth16_1}). It can be seen that there cannot exist another variable node (either good or bad) which receives two incorrect messages without violating the $g=16$ constraint. Also, $v_8$ receives all correct messages and $v_2^1,v_2^2,v_4^1,v_4^2,v_6^1,v_6^2$ receive one incorrect message each from $c_3^1,c_3^2,c_5^1,c_5^2,c_7^1,c_7^2$, respectively. In the fifth iteration, we have 
\begin{eqnarray}
\omega_5(v_0,c_1^3)&=&1, \nonumber  \\
\omega_5(v,:\backslash c)&=&\{1\}, ~~(v,c)\in \{(v_2^1,c_3^1),(v_2^2,c_3^2),\nonumber \\ 
&&~~~~~(v_4^1,c_5^1),(v_4^2,c_5^2),(v_6^1,c_7^1),(v_6^2,c_7^2)\}, \nonumber \\
\omega_5(v,c)&=&0, ~~\mbox{otherwise}. \nonumber 
\end{eqnarray}
In the sixth iteration, we have
\begin{eqnarray}
\omega_6(v_0,c_1^3)&=&1, \nonumber  \\
\omega_6(v,:\backslash c)&=&\{1\}, ~~(v,c)\in \{(v_2^1,c_3^1),(v_2^3,c_3^2),(v_4^1,c_5^1),\nonumber \\
&&~~~~~(v_4^2,c_5^2)\},\nonumber \\ 
\omega_6(v,c)&=&0, ~~\mbox{otherwise}. \nonumber 
\end{eqnarray}
In the seventh iteration, we have,
\begin{eqnarray}
\omega_7(v_0,c_1^3)&=&1, \nonumber  \\
\omega_7(v,:\backslash c)&=&\{1\}, ~~(v,c)\in \{(v_2^1,c_3^1),(v_2^2,c_3^2)\}\nonumber \\ 
\omega_7(v,c)&=&0, ~~\mbox{otherwise}. \nonumber 
\end{eqnarray}
Finally in the eighth iteration, we have,
\begin{eqnarray}
\omega_8(v_0,c_1^3)&=&1, \nonumber  \\
\omega_8(v,c)&=&0, ~~\mbox{otherwise} \nonumber 
\end{eqnarray}

At the end of eighth iteration, no variable node receives two incorrect messages and hence the decoder converges to a valid codeword. 

\end{proof}
\subsubsection{$g=12$}
In this case, $k=2$ and the incoming messages at the end of the second iteration are independent. We need to prove that any code with Tanner graph with $g=12$, can correct all error patterns of weight less than six. Let $v$ be a variable node which sends an incorrect message in the third iteration along edge $e=(v,c)$ given that there are at most $5$ bad variables and $\mathcal{N}^5_{e}$ is a tree. Fig. \ref{iteration4} illustrates different configurations of bad variable nodes in $\mathcal{N}^{4}_{e}$. 

Fig. \ref{girth12} shows all possible configurations of five or less bad variable nodes which do not converge to a codeword at the end of three iterations. However, all the configurations converge to a codeword in six iterations. The proofs for configurations in Fig. \ref{girth12}(a)-(h) are similar to the proof for configuration in Fig. \ref{girth16_1} and are omitted. Since, configuration (i) has only four bad variable nodes, a complete proof for convergence requires considering all possible locations of the fifth bad variable node, but other than that the structure of the proof is identical to that of the proof for the configuration in Fig. \ref{girth16_1}.  It is worth noting that in this case, there exist configurations in which a bad variable receives two incorrect messages at the end of the third iteration. However, all the configurations eventually converge to a codeword.
\begin{figure*}
\begin{centering}
\subfigure[] 
{
    \label{girth12_1}

\includegraphics[height=1.4in]{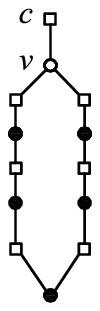}
}
\subfigure[] 
{
    \label{girth12_2}

\includegraphics[height=1.4in]{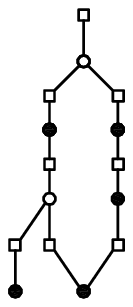}
}
\subfigure[] 
{
    \label{girth12_3}

\includegraphics[height=1.4in]{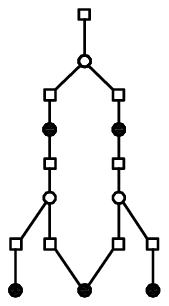}
}
\subfigure[] 
{
    \label{girth12_4}

\includegraphics[height=1.4in]{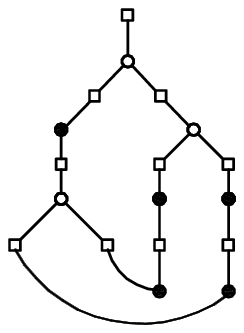}
}
\subfigure[] 
{
    \label{girth12_5}

\includegraphics[height=1.4in]{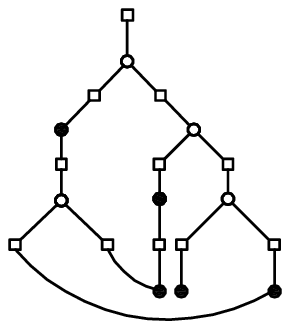}
}
\subfigure[] 
{
    \label{girth12_6}

\includegraphics[height=1.4in]{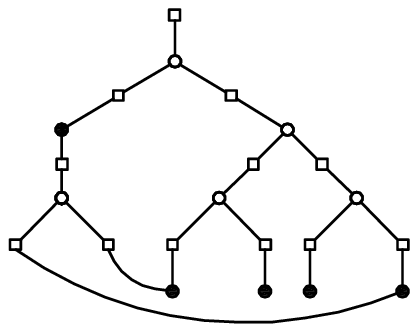}
}
\subfigure[] 
{
    \label{girth12_7}

\includegraphics[height=1.7in]{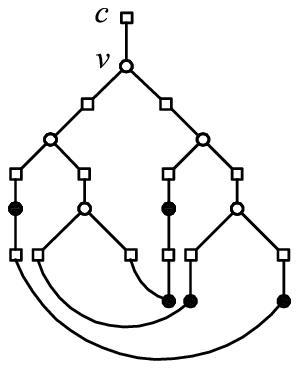}
}
\hspace{0.5in}
\subfigure[] 
{
    \label{girth12_8}

\includegraphics[height=1.7in]{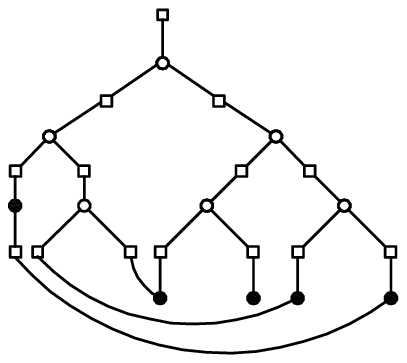}
}
\hspace{0.5in}
\subfigure[] 
{
    \label{girth12_9}

\includegraphics[height=1.7in]{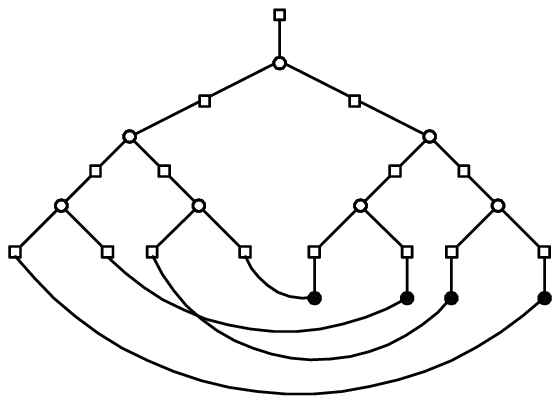}
}
\caption{Configurations of at most 5 variable nodes free of cycles of length less than 12 which do not converge in 3 iterations. }
\label{girth12}
\end{centering}
\end{figure*}

\subsection{$g/2$ is odd}
In this case, we have $k=(g-2)/4$ and we need to prove that the code is capable of correcting all error patterns of weight $g/2-1=2k$ or less. The methodology of the proof is similar to the proof in the case when $g/2$ is even. In this case, we have $\mathcal{N}^{i}_{e_3} \cap \mathcal{N}^{j}_{e_4} = \phi$ for $i+j < 4k-2$. This implies that $\mathcal{N}^{2k}_{e_3}$ and $\mathcal{N}^{2k}_{e_4}$ can have a common node at depth $2k-1$. Therefore, in any configuration, $B(\mathcal{N}^{2k}_{e_3}\cup \mathcal{N}^{2k}_{e_4}) $ is lower bounded by $\max\left(B(\mathcal{N}^{2k-4}_{e_3})+B(\mathcal{N}^{2k}_{e_4}), B(\mathcal{N}^{2k}_{e_3})+ B(\mathcal{N}^{2k-4}_{e_4})\right)$. The valid configurations in this case are the ones which satisfy $B(\mathcal{N}^{2k}_{e_3} \cup \mathcal{N}^{2k}_{e_4}) \leq 2k$. We again deal with $g \geq 14$ and $g=10$ separately.
\subsubsection{$g \geq 14$}
\begin{lema}\label{lemma2}
For $g=14$, there is only one configuration of six bad variable nodes which does not converge in four iterations.
\end{lema}
\begin{proof}
Using arguments outlined above and the configurations in Fig. \ref{iteration3} along with the constraint that $g \geq 14$, we conclude that there is only one configuration which does not converge in four iterations, which is shown in Fig. \ref{girth14}.
\end{proof}
\begin{lema}\label{lemma3}
If $\mathcal{C}_g$ with $g \geq 14$ is a valid configuration which does not converge in $k+1$ iterations, then $v_1^1$ and $v_1^2$ are bad variable nodes
\end{lema}
\begin{proof}
Similar to the proof of Lemma \ref{lemma1}.
\end{proof}
\begin{te}\label{theorem4}
If $\mathcal{C}_g$ is a valid configuration which does not converge in $k+1$ iterations, then there exists a valid configuration $\mathcal{C}_{g-4}$  which does not converge in $k$ iterations.
\end{te}
\begin{proof}
Similar to the proof of Theorem \ref{theorem1}.
\end{proof}
\begin{te}\label{theorem5}
For $k \geq 3$, there is only one valid configuration which does not converge in $k+1$ iterations.
\end{te}
\begin{proof}
For $k=3$, we have $g=14$ and there is only one configuration. For $k =4$, the number of valid configurations cannot be more than one. The valid configuration for $g=14$, can  be extended to a configuration for $g=18$ (in the same manner as in Theorem \ref{theorem2}).  In general, the valid configuration for girth $g$ can be extended to a valid configuration for girth $g+4$. Fig. \ref{girth18} shows $\mathcal{C}_g$ for all $g \geq 14$.
\end{proof}
\begin{te}\label{theorem6}
The configuration $\mathcal{C}_g$ converges to a codeword in $g/2$ iterations.
\end{te}
\begin{proof}
Similar to the proof of Theorem \ref{theorem3}.
\end{proof}
\subsubsection{$g=10$}
In this case, $k=2$ and there are three configurations which do not converge at the end of the third iteration. Fig. \ref{girth10} shows the three configurations. It can be shown that these configurations converge in five iterations.
\section{Discussion}\label{discussion}
In this paper, we have established a relation between the girth and error correction capability of column-weight-three LDPC codes. The result presented is the best possible bound as it is known that codes with girth $g \geq 10$ cannot be guaranteed to correct $g/2$ errors. While it is disappointing that the error correction capability grows only linearly in the girth for column-weight-three codes, the methodology of our proof can be applied to higher column-weight codes to (hopefully) obtain better results.
\begin{figure}
\begin{centering}
\subfigure[] 
{
    \label{girth14}

\includegraphics[width=0.8in]{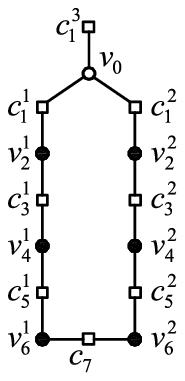}
}
\hspace{0.5in}
\subfigure[] 
{
    \label{girth18}

\includegraphics[width=0.8in]{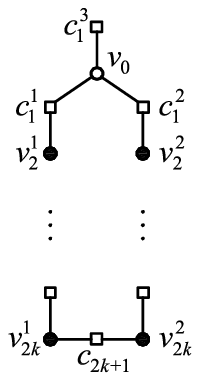}
}
\caption{\subref{girth14} Configuration of at most $6$ bad variable nodes free of cycles of length less than $14$ which does not converge in $4$ iterations \subref{girth18} Configuration of at most $2k$ bad variable nodes free of cycles of length less than $4k+2$ which does not converge in $k+1$ iterations.}
\label{genconfigs}
\end{centering}
\end{figure}


\begin{figure}
\begin{centering}
\subfigure[] 
{
    \label{girth10_1}

\includegraphics[bbllx=0,bblly=0,bburx=50, bbury=115]{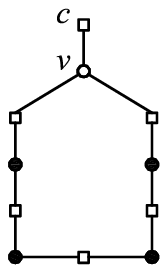}
}
\hspace{0.09in}
\subfigure[] 
{
    \label{girth10_2}

\includegraphics{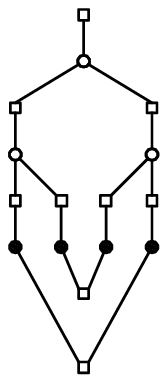}
}
\subfigure[] 
{
    \label{girth10_3}

\includegraphics{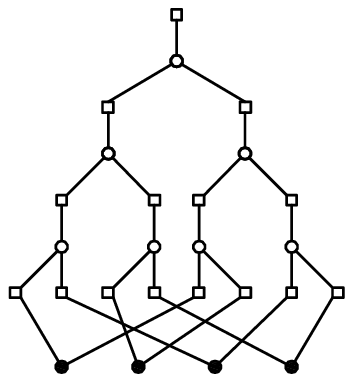}
}
\caption{Configurations of at most $4$ variable nodes free of cycles of length less than $10$ which do not converge in $3$ iterations.}
\label{girth10}
\end{centering}
\end{figure}


\begin{thebibliography}{10}
\providecommand{\url}[1]{#1}
\csname url@rmstyle\endcsname
\providecommand{\newblock}{\relax}
\providecommand{\bibinfo}[2]{#2}
\providecommand\BIBentrySTDinterwordspacing{\spaceskip=0pt\relax}
\providecommand\BIBentryALTinterwordstretchfactor{4}
\providecommand\BIBentryALTinterwordspacing{\spaceskip=\fontdimen2\font plus
\BIBentryALTinterwordstretchfactor\fontdimen3\font minus
  \fontdimen4\font\relax}
\providecommand\BIBforeignlanguage[2]{{%
\expandafter\ifx\csname l@#1\endcsname\relax
\typeout{** WARNING: IEEEtran.bst: No hyphenation pattern has been}%
\typeout{** loaded for the language `#1'. Using the pattern for}%
\typeout{** the default language instead.}%
\else
\language=\csname l@#1\endcsname
\fi
#2}}

\bibitem{gallager}
R.~G. Gallager, \emph{Low Density Parity Check Codes}.\hskip 1em plus 0.5em
  minus 0.4em\relax Cambridge, MA: M.I.T. Press, 1963.

\bibitem{tanner}
R.~M. Tanner, ``A recursive approach to low complexity codes,'' \emph{IEEE
  Trans. Inform. Theory}, vol.~27, no.~5, pp. 533--547, Sept. 1981.

\bibitem{richardsonurbanke}
T.~J. Richardson and R.~Urbanke, ``The capacity of low-density parity-check
  codes under message-passing decoding,'' \emph{IEEE Trans. Inform. Theory},
  vol.~47, no.~2, pp. 599--618, Feb. 2001.

\bibitem{richardsonurbankeshokrollahi}
T.~J. Richardson, M.~Shokrollahi, and R.~Urbanke, ``Design of
  capacity-approaching irregular low-density parity-check codes,'' \emph{IEEE
  Trans. Inform. Theory}, vol.~47, no.~2, pp. 638--656, Feb. 2001.

\bibitem{weaknessmackay}
\BIBentryALTinterwordspacing
D.~J.~C. MacKay and M.~J. Postol, ``Weaknesses of {M}argulis and
  {R}amanujan--{M}argulis low-density parity-check codes,'' in
  \emph{Proceedings of MFCSIT2002, Galway}, ser. Electronic Notes in
  Theoretical Computer Science, vol.~74.\hskip 1em plus 0.5em minus 0.4em\relax
  Elsevier, 2003. [Online]. Available:
  \url{http://www.inference.phy.cam.ac.uk/mackay/abstracts/margulis.html}
\BIBentrySTDinterwordspacing

\bibitem{rich}
T.~J. Richardson, ``Error floors of {LDPC} codes,'' in \emph{Proc. of 41st
  Annual Allerton Conf. on Communications, Control and Computing}, 2003, pp.
  1426--1435.

\bibitem{breakingtrappingsets}
\BIBentryALTinterwordspacing
M.~Ivkovic, S.~K. Chilappagari, and B.~Vasic, ``Eliminating trapping sets in
  low-density parity check codes by using {T}anner graph covers,'' to appear in
  \textit{IEEE Trans. Inform. Theory}. [Online]. Available:
  \url{http://arxiv.org/abs/0805.1662}
\BIBentrySTDinterwordspacing

\bibitem{colwtthreepaper}
\BIBentryALTinterwordspacing
S.~K. Chilappagari and B.~Vasic, ``Error correction capability of
  column-weight-three {LDPC} codes,'' submitted to \textit{IEEE Trans. Inform.
  Theory}. [Online]. Available: \url{http://arxiv.org/abs/0710.3427}
\BIBentrySTDinterwordspacing

\bibitem{itwpaper}
S.~K. Chilappagari, A.~R. Krishnan, and B.~Vasic, ``{LDPC} codes which can
  correct three errors under iterative decoding,'' in \emph{Proc. of IEEE
  Information Theory Workshop}, May 5-9, 2008.

\bibitem{shokrollahi}
A.~Shokrollahi, ``An introduction to low-density parity-check codes,'' in
  \emph{Theoretical aspects of computer science: advanced lectures}.\hskip 1em
  plus 0.5em minus 0.4em\relax New York, NY, USA: Springer-Verlag New York,
  Inc., 2002, pp. 175--197.

\end{thebibliography}
\end{document}